\documentclass[11 pt]{article}
\usepackage[utf8]{inputenc}
\usepackage{amsmath}
\usepackage{amssymb}
\usepackage{amsthm}
\usepackage{xcolor}
\usepackage{float}
\usepackage{booktabs}
\usepackage[shortlabels]{enumitem}
\usepackage[numbers,square,sort&compress]{natbib}
\usepackage{hyperref}
\usepackage[normalem]{ulem}
\usepackage{setspace}
\usepackage{fullpage}
\usepackage{soul}

\newtheorem{thm}{Theorem}[section]
\newtheorem{cor}[thm]{Corollary}
\newtheorem{prop}[thm]{Proposition}
\newtheorem{lem}[thm]{Lemma}

\theoremstyle{definition}

\newtheorem{assump}[thm]{Assumption}

\definecolor{cobalt}{rgb}{0.0, 0.28, 0.67}

\usepackage[dvipsnames]{xcolor}

\theoremstyle{definition}

\title{Can Labor Markets Function in the Age of AI? The Evaluation Bottleneck in Hiring
\\[1em]
\large Preliminary Draft -- Comments Welcome!}
\author{Itai Ashlagi, Ramesh Johari, Jon Kleinberg, Anushka Murthy}
\date{\today}

\begin{document}

\maketitle
\begin{abstract}
AI-assisted job-search tools have become increasingly popular by making it easier to find and apply to jobs. But by making it easier for applicants to generate and tailor application materials, they can also reduce how informative those materials are about applicant fit. We study this tradeoff in a hiring market where applicants differ in experience and latent match quality and firms use noisy application materials to decide whom to screen. We ask how AI affects downstream screening and hiring, and which applicants are most adversely affected. As application materials become less informative, a Bayesian firm rationally relies more heavily on coarse observables such as prior experience. Among the four applicant types defined by experience and compatibility for the job, inexperienced-compatible applicants are the most exposed: they lack observable experience and lose the individualized information that could distinguish them from other inexperienced candidates. When screening is costly, these changes can also generate inefficient screening failures in which firms screen no applicants or screen only experienced applicants. We then show that multistage hiring can arise as an endogenous firm response: a relatively inexpensive intermediate assessment allows firms to acquire new evidence of fit before costly full screening. This can restore screening opportunities that disappear under one-stage hiring and give inexperienced-compatible applicants a path to screening. Our results show how AI can shift the central friction in hiring from submitting applications to obtaining credible evaluation, creating entry barriers for high-fit workers without prior experience. Multistage hiring can endogenously arise in response, restoring evaluation opportunities that would otherwise disappear and helping preserve market functioning.

\end{abstract}
\section{Introduction}

There has been a proliferation of AI-assisted job-search tools. Applicants increasingly use AI systems to draft resumes and cover letters, tailor materials to specific vacancies, search across platforms, and submit applications at scale \citep{softwarefinder2026_ai_job_search,constantino2025_double_edged_sword}. 
Yet the experience of many job seekers and employers has not improved accordingly. Recent accounts describe applicants submitting large numbers of applications with little response, while firms report receiving growing volumes of increasingly similar and difficult-to-verify materials \citep{kearns2025_ai_job_search_quartz,abril2026_ai_resume_wapo,roberthalf2026_ai_applications}. Critics have consequently warned that AI-mediated hiring may push firms toward relying on existing brand-name signals and referral networks rather than a more open and meritocratic labor market \citep{chamorro_premuzic2026_hiring_worse}. This concern is especially salient for new labor-market entrants: recent college graduates continue to face elevated unemployment and underemployment \citep{nyfed2026_college_labor_market}. 

These developments point to a tension in AI-mediated job search. AI can make it easier for workers to enter an applicant pool without necessarily making it easier for firms to determine which applicants are good matches. Indeed, if widespread AI assistance makes written applications less individually informative, the central friction in hiring may shift from submitting an application to the provision and generation of credible, individualized evidence of fit.

Recent evidence suggests that AI weakens applicant-generated signals. \citet{galdin2025making} find that customized applications predicted hiring before the introduction of large language models but became less valuable afterward; eliminating the signaling value of written applications can make hiring substantially less meritocratic. Studying the rollout of an AI-assisted cover-letter tool, \citet{cui2025signaling} find that AI increased tailoring and callbacks but reduced tailoring's predictive content, shifting employers to rely more heavily on past work history and potentially disadvantaging new workers. More generally, \citet{cowgill2026cheapen} show that sender-side access to generative AI reduces screening accuracy on average, although its effects depend on how AI changes the relative informativeness of experts' and non-experts' messages. Together, these studies suggest that AI can weaken applicant-generated signals and shift the evidence firms use for screening.

This paper develops a model of how such a change in the information structure affects access to screening and hiring. While a large literature studies the effects of automation and AI on tasks, productivity, wages, employment, and vacancy creation 
\citep{autor2003skill,autor2015jobs,acemoglu2018race,acemoglu2020robots,webb2020impact,felten2021occupational,brynjolfsson2025generative}, we focus on job search and screening. We consider a firm with a vacancy and heterogeneous potential applicants. Greater AI saturation lowers applicants' costs of applying, reduces the informativeness of application materials, and raises the firm's effective cost of identifying compatible candidates. Applicants differ in observable experience and latent match quality. The firm observes experience and noisy application materials, forms Bayesian posterior beliefs about compatibility, and chooses whom to screen.

Our first set of results shows how degrading individualized information can create an entry barrier for compatible workers without prior experience. As application materials become noisier, the firm  places less weight on applicant-specific evidence and more on experience-based priors. Inexperienced applicants therefore need stronger favorable evidence to overcome experienced applicants’ prior advantage. This mechanism is analogous to classic models of statistical discrimination \citep{phelps1972statistical,aigner1977statistical}, but with respect to applicant experience. 

This reweighting is privately rational but disproportionately harms inexperienced-compatible candidates. As individualized signals deteriorate they lose the evidence that could offset the lower prior associated with inexperience. As a result, increasing AI saturation can harm labor-market entry even for workers who are good matches.

We then show that private screening can be inefficient. The firm screens applicants only when the expected private value of doing so exceeds the screening cost, but it does not internalize the worker's surplus from a successful match. This creates a gap between the firm's privately optimal screening cutoff and the socially efficient screening cutoff, resulting in some applicants who are worth screening being rejected. We characterize two resulting market failures: one in which the firm screens no applicant even though a social planner would screen at least one, and another in which the firm excludes all submitted inexperienced-compatible applicants even though screening at least one of them would increase expected total surplus. As application materials become uninformative, applicants' posterior compatibility converges toward their experience-group priors. These  failures can therefore arise systematically: screening either  shuts down or continues only for experienced applicants while excluding inexperienced-compatible candidates.

Finally, we study whether a multistage hiring process can mitigate these failures. Rather than making hiring decisions entirely from increasingly noisy application materials, firms may create lower-stakes opportunities for applicants to generate additional credible evidence of fit. In our model, this takes the form of a relatively inexpensive intermediate assessment, such as a work-sample task, live skill test, or structured preliminary interview, which allows the firm to acquire additional applicant-specific information before undertaking costly full evaluation. More broadly, internships, predoctoral positions, and other temporary or probationary roles may serve a related function by allowing workers—particularly those with limited prior experience—to demonstrate their match quality through performance.

We characterize an intermediate-assessment region consisting of applicants whose posterior compatibility is too low to justify immediate full screening but high enough that acquiring an additional signal is valuable. These applicants are rejected under one-stage hiring but can advance under multistage hiring after a favorable intermediate assessment. When this region is nonempty, the firm strictly prefers a multistage hiring process, which can avoid the screening failures that occur under one-stage hiring. It also gives inexperienced-compatible applicants who would otherwise be rejected a positive-probability path to screening and hiring. At high levels of AI saturation, this benefit extends to nearly all submitted inexperienced-compatible applicants.

Taken together, our results explain why easier  application need not improve job access or welfare. As applicant-generated materials become less informative, firms rely more on experience-based priors and allocate costly evaluation more selectively, disproportionately harming inexperienced-compatible applicants. Nevertheless, we  find a hopeful feasible path forward: we show that multistage hiring emerges as a privately beneficial firm response that creates a new channel for credible, individualized evidence, and gives inexperienced workers an opportunity to demonstrate fit before firms rely on coarse experience-based priors.

\section{Related Literature}

Our paper connects three streams of research: AI, signaling, and screening in labor markets; directed and sequential search; and statistical discrimination and talent discovery.

\paragraph{AI and labor-market signaling.}
We build on classic models of signaling, screening, and cheap talk
\citep{spence1973jobmarket,stiglitz1975screening,crawford1982strategic,farrell1996cheap}.
Recent evidence shows that generative AI can weaken applicant-generated signals and change the information employers use
\citep{galdin2025making,cui2025signaling,cowgill2026cheapen}.
We study how this deterioration affects access to costly individualized screening and whether multistage evaluation can mitigate the resulting screening failures. Closest to our work, \citet{jungbauer2026} studies how AI-induced deterioration in initial-stage screening signals disrupts assortative matching. Jungbauer also shows that the optimal redesign of a common screening signal need not restore pre-AI levels of informativeness. In contrast, we study how degraded initial information affects access to costly individualized evaluation across experience groups, and how multistage hiring processes can adapt when initial applicant signals become less informative.

\paragraph{Directed and sequential search}
Our model also relates to work on directed search, decentralized matching, and sequential information acquisition. Classic job-search models study costly search under uncertainty \citep{mccall1970economics}, while directed and competitive search models study decentralized markets in which workers direct applications based on observable opportunities and expected competition \citep{peters1991exante,moen1997competitive,burdett2001pricing,wright2021directed,albrecht2006equilibrium,galenianos2009directed}. Related work on online labor markets shows how application costs, congestion, and screening frictions shape applicant and employer behavior \citep{arnosti2021managing,horton2021jobseekers,horton2024reducing,fradkin2025competition,davis2024application}. Most closely, \citet{arnosti2021managing} study a decentralized market in which applicants pay to apply and employers pay to screen. We focus on an additional information friction: the signals firms use to decide whom to screen can themselves become less informative.

The firm's problem also has a sequential-search interpretation. Our one-stage baseline specializes \citet{weitzman1979optimal} to hiring: the firm orders applicants by posterior compatibility and inspects them until it finds an acceptable candidate or further inspection is not worthwhile. The multistage model allows partial information acquisition before full screening, connecting to sequential inspection problems such as \citet{aouad2026pandora} and to hiring models with intermediate assessments or interviews \citep{josephson2016costly,lessem2026matching,sockin2025interviews,jabarianreshidi2025}. We study how the value of this sequential evaluation changes as applicant-generated information deteriorates and experience generates heterogeneous prior beliefs about fit.

\paragraph{Statistical discrimination and entry barriers}

Our entry-barrier results are closely related to statistical discrimination. When firms observe noisy individual signals, they may rationally condition decisions on observable characteristics correlated with expected productivity \citep{phelps1972statistical,arrow1973discrimination,aigner1977statistical}. In our setting, experience plays this role: as applicant-generated materials become less informative, the firm places greater weight on experience-group priors. High-fit inexperienced applicants are especially harmed because they lose the applicant-specific information that could distinguish them from lower-fit workers in the same group.

Related employer-learning models study how firms initially rely on observable characteristics and update as better information arrives \citep{altonji2001employer}. We instead study whether applicants receive the costly information acquisition needed to overcome an unfavorable prior. This connects to work on inefficient talent discovery and certification in entry-level labor markets \citep{pallais2014inefficient,stanton2016landing,tervio2009superstars,li2026hiring}.

Although our baseline is static, reduced access to interviews and jobs may also limit inexperienced workers' ability to accumulate the experience firms increasingly rely on, potentially reinforcing initial disparities \citep{coate1993will,moro2004general,baek2025feedback}. Related to us, \citet{hu2018shortterm} show that temporary two-stage hiring can shift a market toward a more equitable long-run equilibrium; our multistage process instead generates additional applicant-specific information before full screening.

\section{Model}
\label{sec:model}

We study a firm with one vacancy and  $n>1$ potential applicants. AI saturation is modeled by an exogenous parameter $q\in [0,1)$, which affects three screening and application primitives described below: the application-cost scale $c_a(q)$, the variance \(\sigma_m^2(q)\) of the noise in the application-materials signal $m_i$,  and the effective firm screening cost $c_s(q)$.

\paragraph{Environment and payoffs.}
If the firm hires a compatible applicant, the firm receives payoff $v_F>0$. The hired applicant receives payoff $v_A>0$. We interpret $v_F$ and $v_A$ as the respective gains to the firm and applicant from a compatible match, relative to their outside options, so total surplus from a compatible match is $v_F+v_A$.    If no hire occurs, no match surplus is generated; applicants still bear any application costs they have incurred, and the firm still bears any screening costs it has incurred.

\paragraph{Applicants.}
Applicant $i$ has observable experience
\[
e_i\in\{0,1\},
\]
where $e_i=1$ denotes experienced and $e_i=0$ denotes inexperienced. Let
\[
\lambda_e=\Pr(e_i=e),
\qquad
\lambda_0+\lambda_1=1.
\]

Applicant $i$ also has latent match quality $z_i\in\mathbb R$. Match quality is not observed by the firm before screening. Applicant compatibility is given by
\[
k_i=\mathbf 1\{z_i\ge0\}.
\]
The firm seeks to hire an applicant with $k_i=1$. We refer to the pair $(e_i,k_i)\in\{0,1\}^2$ as applicant $i$'s
realized type. Thus $(1,1)$ denotes an experienced-compatible applicant,
$(1,0)$ an experienced-incompatible applicant, $(0,1)$ an
inexperienced-compatible applicant, and $(0,0)$ an
inexperienced-incompatible applicant.

\paragraph{Priors and application materials.}
Conditional on experience,
\begin{equation}
z_i\mid e_i=e
\sim
\mathcal N(\mu_e,\sigma_z^2),
\qquad e\in\{0,1\},
\label{eq:quality-prior}
\end{equation}
with $\mu_1>\mu_0>0$. Thus experienced applicants have higher expected
match quality.

Conditional on applying, applicant $i$ generates a materials signal
\begin{equation}
m_i=z_i+\varepsilon_i,
\qquad
\varepsilon_i\sim\mathcal N(0,\sigma_m^2(q)).
\label{eq:materials-signal}
\end{equation}

\paragraph{Application decisions.}
Applicant $i$ observes her experience $e_i$ and application cost $c_a(q) \eta_i$ before applying, but does not observe $z_i$. Let $d_i\in\{0,1\}$ denote applicant $i$'s application decision. If she applies, she pays
\[
c_a(q)\eta_i.
\]
The cost scalings $\eta_i$ are nonnegative, with $\eta_i\mid e_i=e$ having CDF $F_e$ on $\mathbb R_+$, and capture idiosyncratic differences in the time, effort, or opportunity costs applicants face when preparing and submitting an application. We assume these scalings are independent of match quality and application signal noise conditional on experience. 

Let

$$
\alpha_e(q):=\Pr(d_i=1\mid e_i=e)
$$

denote the equilibrium group-specific application rate. Given application rates $(a_0,a_1)$, let

$$
V_e(q;a_0,a_1)
:=
v_A\Pr(\text{hired}\mid e_i=e,d_i=1;q,a_0,a_1)
$$

denote the expected value of applying for an applicant from experience group $e$. An applicant with experience $e$ applies if and only if

$$
V_e(q;a_0,a_1)\ge c_a(q)\eta_i.
$$

Thus, when $c_a(q)>0$, equilibrium application rates satisfy

$$
\alpha_e(q)
=
F_e\!\left(
\frac{V_e(q;\alpha_0(q),\alpha_1(q))}{c_a(q)}
\right).
$$

When $c_a(q)=0$, we choose a tie-breaking convention that selects

$$
\alpha_e(q)=1.
$$

\paragraph{Firm screening technology.}
The firm observes $(e_i,m_i)$ for each submitted application. Screening applicant $i$ costs
\[
c_s(q)
\]
and reveals $z_i$ perfectly. The firm must screen an applicant before hiring them. 

Let
\[
p_i=\Pr(k_i=1\mid m_i,e_i)
\]
denote applicant $i$'s posterior compatibility. 

\paragraph{Assumptions.}

We complete the model specification by describing how AI saturation affects the key primitives, and by imposing a regularity condition that ensures positive participation by both experience groups.

The following assumption captures the idea that as AI tools become more prevalent, submitted materials become less individually informative about latent match quality.

\begin{assump}[AI saturation lowers signal informativeness]
\label{assump:sigma-increasing}
The signal-noise variance $\sigma_m^2(q)$ is strictly increasing in $q$.
\end{assump}

The next assumption captures the contrasting effects of AI on costs faced by the two sides of the market. AI tools can reduce the time and effort required to generate applications, lowering applicants' effective cost of applying. At the same time, AI may increase the firm's effective cost of identifying compatible applicants, for example because AI-generated application materials can become more homogeneous or less informative.

\begin{assump}[AI saturation lowers application costs and raises effective screening costs]
\label{assump:ca-decreasing}
\label{assump:cs-increasing}
The application-cost scale $c_a(q)$ is strictly decreasing in $q$, and the effective screening cost $c_s(q)$ is strictly increasing in $q$.
\end{assump}

The final assumption on the distribution of experience groups and application costs ensures that both experience groups apply with positive probability; we expand more on this at the end of Section \ref{sec:characterization}.

\begin{assump}[Positive participation assumptions]
\label{assump:positive-participation}
Both experience groups occur with positive probability,
\[
\lambda_e>0,
\qquad e\in\{0,1\},
\]
and zero lies in the support of each experience-specific application-cost distribution:
\[
F_e(x)>0
\qquad
\text{for every }x>0,\ e\in\{0,1\}.
\]
\end{assump}

\paragraph{Timing.}
The timing is as follows.

First, for each applicant $i$, their experience $e_i$, latent match quality $z_i$, and application cost scaling $\eta_i$ are sampled, independently across applicants. Applicants observe $(e_i,\eta_i)$ but do not observe $z_i$. 

Second, applicants decide whether to apply. Applicants who apply pay $c_a(q)\eta_i$.

Third, submitted applications generate materials signals $m_i$, and the firm observes $(e_i,m_i)$ for each applicant in the realized pool of applications.

Fourth, the firm forms posterior compatibility beliefs $p_i=\Pr(k_i=1\mid m_i,e_i)$ and chooses whom to screen. Screening costs $c_s(q)$ per applicant and reveals $z_i$. The firm hires the first compatible applicant it screens, if any, and applicant and firm values are realized.

\section{Bayesian Reweighting and Firm Screening}
\label{sec:characterization}

In this section, we characterize how AI saturation changes the firm's posterior beliefs and optimal one-stage screening decisions. 

\paragraph{Posterior beliefs and Bayesian reweighting.}

Among submitted applicants, the firm updates from the
experience-group prior in Equation \eqref{eq:quality-prior} using the materials signal in Equation
\eqref{eq:materials-signal}.

Let
\[
\Delta_\mu:=\mu_1-\mu_0>0
\]
denote the experience-group prior difference. Define
\[
\kappa(q)
:=
\frac{\sigma_z^2}
{\sigma_z^2+\sigma_m^2(q)}
\]
and
\[
\sigma_{\mathrm{post}}^2(q)
:=
\frac{\sigma_z^2\sigma_m^2(q)}
{\sigma_z^2+\sigma_m^2(q)}.
\]
For each applicant, define the posterior-mean score
\[
S_i(q)
:=
\mathbb E[z_i\mid m_i,e_i;q].
\]
Finally, let
\[
\theta(q)
:=
(1-\kappa(q))\Delta_\mu.
\]
The following result describes how the firm updates their posterior after receiving an application.

\begin{prop}[Bayesian reweighting toward experience]
\label{prop:bayesian-reweighting}
Fix $q$ with $\sigma_m^2(q)>0$. Conditional on a given level of AI saturation and applicant $i$'s materials $m_i$ and experience $e_i$, the firm assigns them a posterior compatibility probability given by:
\[
p_i(q)
=
\Phi\left(
\frac{S_i(q)}
{\sigma_{\mathrm{post}}(q)}
\right),
\]
where
\[
S_i(q)
=
\kappa(q)m_i+(1-\kappa(q))\mu_{e_i}
=
(1-\kappa(q))\mu_0+\theta(q)e_i+\kappa(q)m_i.
\]
For any submitted applicants $i$ and $j$,
\[
p_i(q)\ge p_j(q)
\quad\Longleftrightarrow\quad
S_i(q)\ge S_j(q).
\]

Under Assumption \ref{assump:sigma-increasing}, for any
$\widetilde q>q$,
\[
\kappa(\widetilde q)<\kappa(q),
\qquad
\theta(\widetilde q)>\theta(q),
\]
and
\[
\frac{\theta(q)}{\kappa(q)}
=
\frac{\Delta_\mu\,\sigma_m^2(q)}{\sigma_z^2}
\]
is strictly increasing in $q$.
\end{prop}

The proposition formalizes the statistical-discrimination mechanism.
As application materials become noisier, the firm's Bayes-optimal response is to place less weight on
the applicant-specific materials signal and more weight on the applicant's
prior experience. Equivalently, the amount of materials evidence required
to offset the prior difference between experienced and inexperienced applicants
increases with AI saturation.

\paragraph{Optimal one-stage screening.}

Suppose the firm receives $M$ submitted applications and observes posterior compatibilities
for each applicant in the realized pool. Relabel the submitted applicants so that
\[
p_1(q)\ge p_2(q)\ge\cdots\ge p_M(q).
\]
Define the firm's private screening cutoff by
\[
p^F(q)
:=
\frac{c_s(q)}{v_F}.
\]
The following proposition characterizes the firm's optimal one-stage screening rule.\footnote{This problem is a binary-payoff special case of Weitzman's Pandora problem; see \citep{weitzman1979optimal}. Each applicant is a ``box,'' screening is inspection, and the hidden payoff is $v_F k_i$. To remain self-contained, we give a direct proof of the optimal screening rule.}

\begin{prop}[Optimal one-stage screening rule]
\label{prop:optimal-screening}
Fix $q$ and a realized applicant pool. Suppose $v_F>0$ and consider $p^F(q):=\frac{c_s(q)}{v_F}\in(0,1)$. The firm's optimal screening policy is to order applicants by decreasing posterior compatibility and screen them sequentially until either a compatible applicant is found or no remaining applicant satisfies
\[
p_i(q)\ge p^F(q).
\]
Equivalently, the firm's screening list consists of applicants whose posterior compatibility satisfies
\[
p_i(q)\ge p^F(q)=\frac{c_s(q)}{v_F}.
\]

Under Assumption \ref{assump:cs-increasing}, for every $\widetilde q>q$,
\[
p^F(\widetilde q)>p^F(q).
\]
\end{prop}

The result follows from comparing the marginal value of screening the next applicant to the marginal screening cost. Conditional on reaching applicant $n$, screening is profitable exactly when
\[
v_F p_n\ge c_s(q).
\]
Thus the posterior compatibility cutoff is given by
\[
p^F(q)=\frac{c_s(q)}{v_F}.
\]

Proposition \ref{prop:optimal-screening} implies that the firm can equivalently use a score cutoff. Since
\[
p_i
=
\Phi\left(
\frac{S_i(q)}
{\sigma_{\mathrm{post}}(q)}
\right)
\]
and $\Phi$ is strictly increasing, the posterior cutoff $p_i\ge  p^F(q)$ is equivalent to
\[
S_i(q)\ge s^F(q),
\]
where
\[
s^F(q):=\sigma_{\mathrm{post}}(q)
\Phi^{-1}\left(p^F(q)\right).
\]
Thus the firm may equivalently rank applicants by their posterior-mean scores and include applicant $i$ in its screening list if and only if $S_i(q)\ge s^F(q)$.

\paragraph{Positive application rates.} We end this section by noting that the results in the following sections require only that applicants from both experience groups submit
applications with positive probability. Assumption \ref{assump:positive-participation} ensures that this occurs for any $q$ such that $p^F(q)\in(0,1)$. Because the Gaussian materials signal has full support, an applicant from either
experience group has positive probability of being compatible and drawing
materials that place their posterior above $p^F(q)$. There is also positive
probability that every other potential applicant is incompatible. Hence
\[
V_e(q;\alpha_0,\alpha_1)>0,
\qquad e\in\{0,1\},
\]
and therefore
\[
\alpha_e(q)>0,
\qquad e\in\{0,1\}.
\]

\section{Impact on Workforce Entry and Market Failure}
\label{sec:workforce-entry}

Section \ref{sec:characterization} shows that noisier application materials shift the firm toward experience-based priors and change its screening decisions. We now show that the effects of this reweighting are uneven across the applicant pool: inexperienced-compatible applicants are especially exposed, and sufficiently high AI saturation can generate screening failures in which the firm forgoes screening anyone or screens only experienced candidates.

\subsection{Shortlisting Exposure}
\label{subsec:shortlisting-exposure}

Recall that applicants have realized types
\[
(e,k)\in\{0,1\}^2,
\]
where $e$ denotes experience and $k$ denotes compatibility. Although compatibility is observed only after screening, conditioning on realized compatibility allows us to identify which types are helped or harmed by the firm's Bayes-optimal scoring rule. By Proposition \ref{prop:bayesian-reweighting}, the firm's posterior-mean score
is
\[
S_i(q)
=
\mathbb E[z_i\mid m_i,e_i;q]
=
\kappa(q)m_i+(1-\kappa(q))\mu_{e_i}.
\]

Higher AI saturation lowers $\kappa(q)$, so the firm places less weight on application materials and more weight on experience. We first study how this reweighting changes the average score of each realized type. We then study how higher AI saturation affects the probability that an applicant clears the firm's posterior screening threshold.

\paragraph{Conditional mean score.} For applicants who apply, define the conditional mean score of type $(e,k)$ by
\[
\bar S_{ek}(q)
:=
\mathbb E[
S_i(q)
\mid
e_i=e,\ k_i=k,\ d_i=1
].
\]

Because applicants do not observe $z_i$ before applying and the application costs are independent of match quality conditional on experience, conditioning on application does not change the distribution of $z_i$ within an experience group.  Hence,
\[
\bar S_{ek}(q)
=
\kappa(q)\mu_{ek}
+
(1-\kappa(q))\mu_e,
\]
where
\[
\mu_{ek}
:=
\mathbb E[z_i\mid e_i=e,k_i=k].
\]
For AI saturation levels $q,\tilde{q}$, let $\Delta\bar S_{ek}(q,\widetilde q):=\bar S_{ek}(\tilde{q})-\bar S_{ek}(q)$ denote the change in conditional mean score.

\begin{prop}[Conditional mean score effects]
\label{prop:mean-score-effects}
Consider AI saturation levels $q<\widetilde q$. Under Assumption
\ref{assump:sigma-increasing},
\[
\Delta\bar S_{ek}(q,\widetilde q)
=
\left[
\kappa(\widetilde q)-\kappa(q)
\right]
\left(
\mu_{ek}-\mu_e
\right).
\]
Moreover,
\[
\Delta\bar S_{01}(q,\widetilde q)
<
\Delta\bar S_{11}(q,\widetilde q)
<
0
<
\Delta\bar S_{00}(q,\widetilde q)
<
\Delta\bar S_{10}(q,\widetilde q).
\]
\end{prop}

Proposition \ref{prop:mean-score-effects} shows that among the four realized types, inexperienced-compatible applicants face the largest negative change in their conditional mean score. In particular, Bayesian reweighting toward experience creates an entry barrier for new but high-fit candidates. 

The intuition behind the result is that conditional
on experience, compatible applicants have above-average match quality: $\mu_{e1}>\mu_e$. As $\kappa(q)$ falls, their favorable individualized information receives less
weight, so their conditional mean scores fall. Incompatible applicants have
below-average match quality so reducing the weight on individualized information raises their mean posterior
scores. Within the pool of compatible applicants, inexperienced-compatible applicants experience a greater decrease than experienced-compatible applicants because their experience-group prior has a lower mean. In other words, they don't have pre-existing experience to compensate for the decreased weight given to their application materials.

\paragraph{Probability of clearing the screening threshold.} 
The mean-score result captures how AI saturation shifts the average position of each realized type in the firm's screening score. However, inclusion in the
firm's screening list is instead a threshold event: applicant $i$ is included
only if
\[
p_i(q)\ge p^F(q).
\]
We therefore next study
the probability that a type-$(e,k)$ applicant clears the shortlisting score threshold. To separate the effect of declining signal informativeness from movement in the firm's screening threshold, we study the standardized match quality and application materials across experience groups.
 
Fix an experience group $e$ and define
\[
X_i
:=
\frac{z_i-\mu_e}{\sigma_z},
\qquad
Y_i(q)
:=
\frac{m_i-\mu_e}
{\sqrt{\sigma_z^2+\sigma_m^2(q)}}.
\]
Conditional on $e_i=e$, both variables are standard normal, with correlation
\[
r(q)
:=
\frac{\sigma_z}
{\sqrt{\sigma_z^2+\sigma_m^2(q)}}.
\]
Thus $r(q)$ summarizes how informative application materials are about match
quality. Under Assumption \ref{assump:sigma-increasing}, $r(q)$ is strictly
decreasing in $q$.

Compatibility is equivalent to
\[
X_i\ge a_e,
\qquad
a_e
:=
-\frac{\mu_e}{\sigma_z}.
\]
Since $\mu_1>\mu_0$, we have $a_1<a_0$.

The firm's screening cutoff can also be expressed in units of the standardized
materials signal. Since
\[
S_i(q)
=
\mu_e+\sigma_z r(q)Y_i(q),
\]
and
\[
p_i(q)
=
\Phi\left(
\frac{S_i(q)}
{\sigma_{\mathrm{post}}(q)}
\right),
\]
the condition
\[
p_i(q)\ge p^F(q)
\]
is equivalent to
\[
Y_i(q)\ge b_e(q),
\]
where
\[
b_e(q)
:=
\frac{
\sigma_{\mathrm{post}}(q)
\Phi^{-1}\left(p^F(q)\right)
-\mu_e
}{
\sigma_z r(q)
}.
\]
Thus $b_e(q)$ is the number of within-group standard deviations by which an
applicant's materials must exceed their experience-group mean to enter the
firm's screening list.

For $b\in\mathbb R$ and $r\in(0,1)$, let $(X,Y)$ be a standard bivariate
normal pair with correlation $r$, and define
\[
\psi_{e1}(b,r)
:=
\Pr(Y\ge b\mid X\ge a_e),
\]
and
\[
\psi_{e0}(b,r)
:=
\Pr(Y\ge b\mid X<a_e).
\]
The screening-list inclusion probability of a submitted type-$(e,k)$ applicant
is therefore
\[
\Pr\left(
p_i(q)\ge p^F(q)
\mid
e_i=e,\ k_i=k,\ d_i=1
\right)
=
\psi_{ek}\left(b_e(q),r(q)\right).
\]

\begin{prop}[Signal informativeness and relative experience thresholds]
\label{prop:two-adverse-shortlisting-channels}
Suppose $p^F(q)\in(0,1)$, $\sigma_m^2(q)>0$, and
$\sigma_m^2(q)$ is differentiable with
\[
\frac{d\sigma_m^2(q)}{dq}>0.
\]
For each $e\in\{0,1\}$, $b\in\mathbb R$, and $r\in(0,1)$,
\[
\frac{\partial\psi_{e1}(b,r)}{\partial r}>0,
\qquad
\frac{\partial\psi_{e0}(b,r)}{\partial r}<0.
\]
Moreover,
\[
b_0(q)-b_1(q)
=
\frac{\Delta_\mu}
{\sigma_z r(q)},
\]
and hence
\[
\frac{d}{dq}
\left[
b_0(q)-b_1(q)
\right]
=
-
\frac{\Delta_\mu r'(q)}
{\sigma_z r(q)^2}
>0.
\]
\end{prop}

The first part of Proposition
\ref{prop:two-adverse-shortlisting-channels} isolates the effect of increased AI saturation on threshold clearing probability through changing signal informativeness while holding the standardized screening threshold fixed.
Because increasing AI saturation reduces signal informativeness, i.e., $r'(q)<0$, noisier application materials reduce the probability that a
compatible applicant clears a fixed threshold and increase this probability for an incompatible applicant. The second part identifies how AI saturation affects the threshold clearing probability by changing the relative experience-group thresholds. The standardized
threshold for inexperienced applicants is higher compared to experienced applicants,
and this difference grows as application materials become less informative. 

Inexperienced-compatible applicants are therefore the only group on the adverse side of both of these mechanisms: compatibility makes them vulnerable to the decline in signal
informativeness, while inexperience requires them to clear a higher relative screening threshold compared to experienced applicants.

\subsection{Market Failure}
\label{subsec:market-failure}

The preceding results show that inexperienced-compatible applicants are especially
exposed to the loss of individualized information at the shortlisting stage. We
now show that the firm's privately optimal screening decision can also be
socially inefficient.

The source of inefficiency is that the firm does not internalize the worker's surplus from a successful match. The firm receives value $v_F$ from hiring a compatible applicant, while total surplus from a compatible match is $v_F+v_A$. Recall from Proposition \ref{prop:optimal-screening} that the firm's private
screening cutoff is
\[
p^F(q)
=
\frac{c_s(q)}{v_F}.
\]
Define the following screening cutoff; in the next proposition, we show it is the cutoff a social planner would use:
\[
p^P(q)
:=
\frac{c_s(q)}{v_F+v_A}.
\]
We suppose throughout this subsection that $v_F>0$, $v_A>0$, and
$p^F(q)\in(0,1)$. Then it follows that
\[
0<p^P(q)<p^F(q)<1.
\]

\begin{prop}[Firm and social planner screening cutoffs]
\label{prop:private-social-cutoffs}
Fix $q$ and a realized applicant pool. The social planner's optimal screening rule is to order applicants by
decreasing posterior compatibility and include applicant $i$ in the screening
list if and only if
\[
p_i(q)\ge p^P(q).
\]
Consequently, applicant $i$ is included in the planner's screening list but not
the firm's screening list if and only if
\[
p^P(q)\le p_i(q)<p^F(q).
\]
\end{prop}

The difference between the private and social screening cutoffs is
\[
p^F(q)-p^P(q)
=
\frac{v_Ac_s(q)}
{v_F(v_F+v_A)}
>0.
\]
Thus the preceding shows that there is a nonempty range of applicants whom the planner would screen but
the firm would reject. Under Assumption \ref{assump:cs-increasing}, the gap between the private and social screening cutoffs increases with AI saturation:

\[
\frac{d}{dq}
\left[
p^F(q)-p^P(q)
\right]
=
\frac{v_Ac_s'(q)}
{v_F(v_F+v_A)}
>0.
\]

Higher effective screening costs therefore expand the range of posterior
compatibilities for which screening is socially valuable but privately
unprofitable.

\paragraph{Two screening-failure events.}

We first characterize two screening-failure events conditional on the firm's realized
pool of submitted applications. We eventually show that these screening-failure events occur with probability approaching one in the high-AI limit. Let
$\mathcal A(q)$ denote the submitted applicant pool. When
$\mathcal A(q)$ is nonempty, define
\[
p^{\max}(q)
:=
\max_{i\in\mathcal A(q)}p_i(q).
\]

This is the highest posterior compatibility among submitted applications.

\begin{prop}[No-hire screening failure]
\label{prop:no-hire-market-failure}
Suppose $\mathcal A(q)$ is nonempty and
\[
p^P(q)<p^{\max}(q)<p^F(q).
\]
Then the firm's screening list is empty, while the social planner's screening
list is nonempty.
\end{prop}

Under the conditions of this result, the firm screens no applicants since
$p^{\max}(q)<p^F(q)$ and hence makes no hire. However, a social planner would
screen at least the applicant with posterior compatibility $p^{\max}(q)$ and
would hire that applicant if they are compatible. Thus the firm does not screen
anyone even though screening and potential hiring would increase expected total
surplus.

A more targeted failure occurs when the firm excludes all
inexperienced-compatible applicants from screening even though at least one such
applicant is socially worth screening. Let
\[
\mathcal A_{01}(q)
:=
\left\{
i\in\mathcal A(q):
e_i=0,\ k_i=1
\right\}
\]
denote the set of submitted inexperienced-compatible applicants. When
$\mathcal A_{01}(q)$ is nonempty, let
\[
p_{01}^{\max}(q)
:=
\max_{i\in\mathcal A_{01}(q)}p_i(q)
\]
be the highest posterior compatibility across the inexperienced-compatible applicants.

\begin{prop}[Inexperienced-compatible screening failure]
\label{prop:inexperienced-compatible-market-failure}
Suppose $\mathcal A_{01}(q)$ is nonempty and
\[
p^P(q)<p_{01}^{\max}(q)<p^F(q).
\]
Then no inexperienced-compatible applicant is included in the firm's screening
list, while at least one inexperienced-compatible applicant is included in the
social planner's screening list.
\end{prop}

Under the conditions of Proposition
\ref{prop:inexperienced-compatible-market-failure}, inexperienced-compatible
candidates can therefore be fully excluded from screening even though screening
and potentially hiring an inexperienced-compatible candidate would increase
expected total surplus.

The preceding two propositions characterize the conditions under which the two screening failures occur. On their own, however, these
fixed-pool conditions do not imply that the failures are systematically likely to arise. We therefore turn next to the
high-AI limit, where we show that the same failure events can occur with
probability approaching one as individualized information disappears.

\paragraph{Screening failures in the high-AI limit.}

We now turn to the central result of this subsection. Although the preceding
propositions characterize when the two failures occur for a realized pool, the
high-AI limit shows that these failures can occur with probability approaching one. As application materials lose their informativeness, applicants'
posterior compatibility probabilities collapse toward their experience-group
priors. Screening therefore increasingly depends on coarse group-level
information rather than applicants' realized compatibility, allowing the
private and social screening thresholds to generate persistent screening
failures.

For each experience group, define the prior compatibility probability
\[
\pi_e
:=
\Pr(k_i=1\mid e_i=e)
=
\Phi\left(
\frac{\mu_e}{\sigma_z}
\right).
\]
Because $\mu_1>\mu_0$, $\pi_1>\pi_0$. Consider a sequence $\{q_n\}$ such that
\[
\sigma_m^2(q_n)\to\infty
\]
and
\[
p^F(q_n)\to p_\infty^F\in(0,1).
\]
Since
\[
p^P(q)
=
\frac{v_F}
{v_F+v_A}
p^F(q),
\]
the social cutoff converges to
\[
p_\infty^P
:=
\frac{v_F}
{v_F+v_A}
p_\infty^F,
\]
where $0<p_\infty^P<p_\infty^F<1$. The following lemma shows that when application
materials become uninformative, individual posterior compatibility converges to
the prior associated with the applicant's experience group. 

\begin{lem}[High-noise posterior pooling]
\label{lem:high-noise-posterior-limit}
For any fixed submitted applicant with realized type $(e,k)$,
\[
p_i(q_n)
\to
\pi_e
\]
in probability, conditional on
\[
e_i=e,\qquad k_i=k,\qquad d_i=1.
\]
Consequently, for any fixed finite submitted pool $\mathcal B$,
\[
\max_{i\in\mathcal B}
\left|
p_i(q_n)-\pi_{e_i}
\right|
\to0
\]
in probability, conditional on the realized types in $\mathcal B$.
\end{lem}

Note that the convergence in the above lemma holds even after conditioning on an applicant's realized compatibility: the
firm cannot recover that compatibility from increasingly noisy materials.

Now fix a finite submitted pool $\mathcal B$, viewed as given at the screening
stage, and condition on its realized applicant types. Define
\[
\pi_{\mathcal B}^{\max}
:=
\max_{i\in\mathcal B}\pi_{e_i}.
\]
Also define
\[
\mathcal B_{01}
:=
\left\{
i\in\mathcal B:
e_i=0,\ k_i=1
\right\},
\qquad
\mathcal B_1
:=
\left\{
i\in\mathcal B:
e_i=1
\right\}.
\]

\begin{cor}[High-noise screening failures]
\label{cor:high-noise-screening-failures}
The following statements hold along the sequence $\{q_n\}$.

\begin{enumerate}[(i)]
    \item If $p_\infty^P<\pi_{\mathcal B}^{\max}<p_\infty^F$, then
    \[
    \Pr\left(
    p^P(q_n)
    <
    \max_{i\in\mathcal B}p_i(q_n)
    <
    p^F(q_n)
    \right)
    \to1.
    \]

    \item Suppose $\mathcal B_{01}$ is nonempty. If $p_\infty^P<\pi_0< p_\infty^F$, then
    \[
    \Pr\left(
    p^P(q_n)
    <
    \min_{i\in\mathcal B_{01}}p_i(q_n)
    \le
    \max_{i\in\mathcal B_{01}}p_i(q_n)
    <
    p^F(q_n)
    \right)
    \to1.
    \]

    If, in addition, $\mathcal B_1$ is nonempty and $p_\infty^F<\pi_1$, then
    \[
    \Pr\left(
    \min_{i\in\mathcal B_1}p_i(q_n)
    >
    p^F(q_n)
    >
    \max_{i\in\mathcal B_{01}}p_i(q_n)
    \right)
    \to1.
    \]
\end{enumerate}
\end{cor}

Part (i) implies that, conditional on the submitted pool $\mathcal B$, the
no-hire screening failure in Proposition
\ref{prop:no-hire-market-failure} occurs with probability approaching one.
Part (ii) implies that every submitted inexperienced-compatible applicant lies
below the firm's cutoff but above the social cutoff with probability approaching
one. Thus the inexperienced-compatible screening failure in Proposition
\ref{prop:inexperienced-compatible-market-failure} occurs with probability
approaching one. Under the additional condition
\[
\pi_0
<
p_\infty^F
<
\pi_1,
\]
the firm continues to include experienced applicants in its screening list while
excluding all submitted inexperienced-compatible applicants. The limiting
outcome is therefore a screening failure concentrated on inexperienced
candidates rather than a complete shutdown of screening.

Taken together, these results show that the screening failures identified above
can arise with probability approaching one as AI saturation grows. As
application materials become uninformative, the firm can no longer distinguish
compatible from incompatible applicants within an experience group, and
posterior beliefs collapse toward group-level priors. The hiring process is
therefore governed increasingly by observable experience rather than realized
compatibility. Depending on the location of the experience-group priors relative
to the private and social screening cutoffs, this produces either a complete
shutdown of privately provided screening or continued screening of experienced
applicants alongside the systematic exclusion of inexperienced-compatible
applicants. 

\section{Multistage Hiring as a Response to AI-Induced Entry Barriers}
\label{sec:multistage}

The preceding section shows that AI saturation can raise entry barriers by making written application materials less informative and full screening more expensive. This section studies a firm response: adding a cheaper intermediate assessment before costly full screening. Examples include live skill tests, work-sample exercises, or other verifiable tasks that are harder for AI-generated application materials to mimic.

We show that the firm strictly benefits from adding an intermediate screening stage and that doing so can restore hiring opportunities for inexperienced-compatible applicants who are disadvantaged by degraded initial signals. These benefits become especially strong as AI saturation grows.

\subsection{Multistage Screening } 
\label{subsec:multistage-technology} 

After observing $(m_i,e_i)$ and forming the baseline posterior
\[
p_i(q)
:=
\Pr(k_i=1\mid m_i,e_i;q),
\]
the firm may pay an intermediate-assessment cost
\[
c_\ell(q)\in[0,c_s(q))
\]
to observe an additional signal $a_i$. After observing this signal, the firm's
posterior compatibility belief is
\[
\widetilde p_i(q)
:=
\Pr(k_i=1\mid m_i,e_i,a_i;q).
\]
The firm may then reject the applicant or pay the full screening cost $c_s(q)$,
which reveals compatibility perfectly.

Submitted applications remain available until the vacancy is filled or the firm
terminates search, and retaining an unassessed application in the pool is
costless. Let
\[
\Pi^{1S}(q)
\qquad\text{and}\qquad
\Pi^{MS}(q)
\]
denote the sets of feasible one-stage and multistage policies, respectively,
conditional on the realized submitted pool $\mathcal A(q)$ and the pre-assessment information $\left\{(m_i,e_i):i\in\mathcal A(q)\right\}$ observed by the firm. For any feasible policy $\pi$, let $U(\pi;q)$ denote the firm's
expected payoff under $\pi$, conditional on this information and net of all
intermediate-assessment and full-screening costs. Define
\[
U^{1S,*}(q)
:=
\sup_{\pi\in\Pi^{1S}(q)}U(\pi;q)
\]
and
\[
U^{MS,*}(q)
:=
\sup_{\pi\in\Pi^{MS}(q)}U(\pi;q).
\]

Because the firm can implement any one-stage policy by never using the
intermediate assessment,
\[
\Pi^{1S}(q)\subseteq\Pi^{MS}(q).
\]
Consequently,
\begin{equation}
U^{MS,*}(q)\ge U^{1S,*}(q).
\label{eq:multistage-weak-dominance}
\end{equation}
This inequality follows directly from nesting. The results in the next subsection therefore focus
on conditions under which the inequality is strict and the intermediate
assessment is used with positive probability.

The value of the intermediate assessment comes from its ability to refine the
firm's posterior before the full screening cost is incurred. By iterated
expectations,
\[
\mathbb E\left[
\widetilde p_i(q)
\mid m_i,e_i;q
\right]
=
p_i(q).
\]
Thus the refined posterior is a mean-preserving refinement of the baseline
posterior.

To isolate the value of this refinement, suppose the vacancy remains open and
applicant $i$ is the firm's only remaining option. If the firm's current
posterior compatibility belief is $x$, its optimal continuation payoff without
an intermediate assessment is
\[
f_q(x)
:=
\left(v_F x-c_s(q)\right)_+
=
v_F\left(x-p^F(q)\right)_+,
\]
where
\[
(y)_+:=\max\{y,0\}
\]
and
\[
p^F(q)=\frac{c_s(q)}{v_F}
\]
is the one-stage screening cutoff from Proposition
\ref{prop:optimal-screening}.

Because $f_q$ is convex, Jensen's inequality implies
\[
\mathbb E\left[
f_q\left(\widetilde p_i(q)\right)
\mid m_i,e_i;q
\right]
\ge
f_q\left(p_i(q)\right).
\]
This is the standard ``value-of-information" \citep{blackwell1953equivalent}: before accounting for the cost
of acquiring it, an additional posterior refinement cannot lower the firm's
optimal expected continuation payoff.

The firm administers the intermediate assessment to applicant $i$ whenever the expected value of the additional information is at least as large as the assessment cost:
\[
\mathbb E\left[
f_q\left(\widetilde p_i(q)\right)
\mid m_i,e_i;q
\right]
-
f_q\left(p_i(q)\right)
\ge c_\ell(q).
\]
The next subsection specializes the assessment to a binary signal and characterizes the applicants whose posterior compatibility lies below the one-stage hiring cutoff but for whom the firm would strictly prefer to administer the assessment if the applicant were the firm's only remaining option.

\subsection{A Binary Assessment and Multistage Screening}
\label{subsec:binary-assessment-middle-band}

To identify which below-cutoff applicants can benefit from the additional stage,
for simplicity we consider a binary intermediate assessment with outcome
\[
A_i\in\{H,L\},
\]
where $H$ is favorable and $L$ is unfavorable.

For each experience group
$e\in\{0,1\}$, assume
\[
\Pr(A_i=H\mid k_i=1,e_i=e,m_i;q)=\tau_e,
\]
and
\[
\Pr(A_i=H\mid k_i=0,e_i=e,m_i;q)=\ell_e,
\]
where
\[
0\le \ell_e<\tau_e\le1.
\]
Thus, conditional on experience and compatibility, application materials do not
provide additional information about the assessment outcome. The parameters
$\tau_e$ and $\ell_e$ may differ across experience groups but $\tau_e>\ell_e$ requires that within each experience group, compatible applicants are more likely to pass.

For an experience-$e$ applicant with pre-assessment posterior compatibility
$p\in(0,1)$, Bayes' rule gives the posterior compatibility probability following a favorable outcome:
\[
p_H^e(p)
:=
\frac{\tau_e p}
{\tau_e p+\ell_e(1-p)},
\]
and the posterior following an unfavorable outcome:
\[
p_L^e(p)
:=
\frac{(1-\tau_e)p}
{(1-\tau_e)p+(1-\ell_e)(1-p)}.
\]
Because $\tau_e>\ell_e$,
\[
p_H^e(p)>p>p_L^e(p).
\]

We next consider the continuation problem in which the vacancy remains open and
an applicant is the firm's only remaining option. Recall that the firm's
one-stage screening cutoff is
\[
p^F(q)
=
\frac{c_s(q)}{v_F}.
\]

For each experience group $e$, define the intermediate-assessment cutoff
\begin{align}
\label{eq: int cutoff}
\underline p_e^A(q)
:=
\frac{
c_\ell(q)/v_F+\ell_e p^F(q)
}{
\tau_e\bigl(1-p^F(q)\bigr)+\ell_e p^F(q)
}.
\end{align}

\begin{lem}[Binary intermediate-assessment cutoff]
\label{lem:binary-assessment-cutoff}
Fix $q$ such that $p^F(q)\in(0,1)$ and fix an experience group
$e\in\{0,1\}$. Consider an applicant with pre-assessment posterior
$p<p^F(q)$ who is the firm's only remaining option.

The firm weakly prefers to administer the intermediate assessment and proceed
to full screening only following $H$ if and only if
\[
p\ge \underline p_e^A(q).
\]
The preference is strict if the inequality above is strict.
\end{lem}

This lemma suggests a simple continuation rule for screening: if the vacancy remains open and a remaining applicant has posterior compatibility
\[
p \in (\underline p_e^A(q), p^F(q)),
\]
then the firm administers the intermediate assessment. After a favorable outcome
\(H\), the applicant's posterior rises above the full-screening cutoff, so the
firm proceeds to full screening. After an unfavorable outcome \(L\), the applicant
remains below the cutoff and is rejected.

This rule does not characterize the globally optimal order of assessment
and screening when multiple applicants remain. In that problem, the firm's
decision may depend on the continuation values generated by its other options.
Nevertheless, this rule is useful because it identifies a set of below-cutoff applicants who are valuable whenever they become the only remaining option. It also provides a feasible multistage policy that can be used to prove
strict adoption of the multistage technology.

\begin{cor}[Strict ex ante value of multistage hiring]
\label{cor:single-firm-adopts-ms}
Suppose $\sigma_m^2(q)>0$ and that, for some $e\in\{0,1\}$,
\[
\lambda_e\alpha_e(q)>0
\]
and
\[
c_\ell(q)
<
v_F p^F(q)\bigl(1-p^F(q)\bigr)
(\tau_e-\ell_e).
\]
Then
\[
\mathbb E\left[
U^{MS,*}(q)-U^{1S,*}(q)
\right]
>0.
\]
Consequently, the firm strictly prefers access to the multistage technology ex
ante, and every ex ante optimal multistage policy administers the intermediate
assessment with positive probability.
\end{cor}

To establish Corollary \ref{cor:single-firm-adopts-ms}, observe that the condition on $c_\ell(q)$ in the statement of the corollary implies
\[
\underline p_e^A(q)<p^F(q).
\]
Because the Gaussian application signal has full support, conditional on
experience $e$ and submission, posterior compatibility places positive
probability on
\[
\left(
\underline p_e^A(q),p^F(q)
\right).
\]
Thus there is positive probability that the submitted pool contains an
experience-$e$ applicant who is rejected under one-stage hiring but is strictly
worth assessing if reached.

A feasible multistage policy can first replicate the optimal one-stage policy.
If no hire occurs and such an applicant remains available, the firm can then
administer the intermediate assessment. This policy strictly improves the firm's
payoff on a positive-probability event. Since the globally optimal multistage
policy performs at least as well as this feasible policy, access to the
intermediate assessment has strictly positive ex ante value.

\subsection{Implications for Screening Failure}
\label{subsec:multistage-screening-failure}

Section \ref{subsec:market-failure} showed that applicants satisfying
\[
p^P(q)<p_i(q)<p^F(q)
\]
are socially worth screening but are rejected under one-stage hiring. Lemma
\ref{lem:binary-assessment-cutoff} shows that under multistage hiring, an experience-$e$ applicant is strictly worth assessing when they are the firm's only remaining option
whenever
\[
\underline p_e^A(q)<p_i(q)<p^F(q).
\]
Thus, for applicants in experience group $e$, multistage hiring creates a
path to screening for those satisfying
\[
\max\left\{
p^P(q),\underline p_e^A(q)
\right\}
<
p_i(q)
<
p^F(q).
\]
In particular, if
\[
\underline p_e^A(q)\le p^P(q),
\]
then every experience-$e$ applicant who is strictly socially worth screening
but rejected under one-stage hiring is strictly worth assessing if they become
the firm's only remaining option.

\begin{cor}[Multistage hiring mitigates one-stage screening failures]
\label{cor:multistage-avoids-failures}
Fix $q$ and a realized submitted pool.

\begin{enumerate}[(i)]
    \item Suppose the one-stage process exhibits the no-hire screening failure
    in Proposition \ref{prop:no-hire-market-failure}. If there exists a
    submitted applicant $i$ such that
    \[
    \max\left\{
    p^P(q),\underline p_{e_i}^A(q)
    \right\}
    <
    p_i(q)
    <
    p^F(q),
    \]
    then
    \[
    U^{MS,*}(q)>U^{1S,*}(q)=0.
    \]
    Consequently, an optimal multistage policy administers the intermediate
    assessment with positive probability and proceeds to full screening with
    positive probability.

    \item Suppose the one-stage process exhibits the
    inexperienced-compatible screening failure in Proposition
    \ref{prop:inexperienced-compatible-market-failure}. If
    \[
    p_{01}^{\max}(q)>\underline p_0^A(q),
    \]
    then there exists an applicant $i\in\mathcal A_{01}(q)$ satisfying
    \[
    p^P(q)<p_i(q)<p^F(q)
    \]
    and
    \[
    p_i(q)>\underline p_0^A(q).
    \]
    At any history at which the vacancy remains open and applicant $i$ is the
    firm's only remaining option, an optimal multistage policy administers the
    intermediate assessment and proceeds to full screening following $H$.
\end{enumerate}
\end{cor}

The corollary above shows that multistage hiring can mitigate the market failures where no applicants are screened or where no inexperienced-compatible applicants are screened. 

For part (i), the firm can improve strictly on one-stage hiring by assessing the
applicant identified in the statement and otherwise terminating search. By
Lemma \ref{lem:binary-assessment-cutoff}, this feasible policy has strictly
positive expected payoff, whereas the optimal one-stage policy screens no one
and receives zero. Thus introducing the intermediate assessment restores screening that would otherwise not occur.

Part (ii) establishes a different form of mitigation. When an inexperienced-compatible applicant lies above the assessment cutoff but below the firm's one-stage screening cutoff, the applicant is excluded under one-stage hiring but becomes worth assessing if reached while the vacancy remains open. The intermediate assessment therefore creates a route to full screening that is absent under the one-stage process. The next subsection shows that such a history occurs with positive ex ante probability, implying a strictly positive gain in the applicant's hiring probability.

\subsection{Implications for Inexperienced-Compatible Applicants}
\label{subsec:multistage-inexperienced-compatible}

We now quantify the benefit of multistage hiring for
inexperienced-compatible applicants. These applicants are high-fit matches, but they are precisely the
group most exposed to the adverse effects of AI on screening in one-stage hiring. 

Let \(h_i^{1S}\) denote the indicator that applicant \(i\) is hired under the
one-stage technology, and let \(h_i^{MS,*}\) denote the indicator that applicant
\(i\) is hired under a globally optimal multistage policy. Probabilities are
evaluated ex ante over the other potential applicants, their application
decisions, the realized application pool, the firm's screening and assessment decisions, and
any intermediate-assessment signals.

For a submitted inexperienced-compatible applicant with baseline posterior \(p\),
define the hiring-probability gain from optimal multistage hiring by
\[
\begin{aligned}
G_{01}^*(p;q)
:={}&
\Pr\!\left(
h_i^{MS,*}=1
\mid e_i=0,\ k_i=1,\ d_i=1,\ p_i=p
\right) \\
&-
\Pr\!\left(
h_i^{1S}=1
\mid e_i=0,\ k_i=1,\ d_i=1,\ p_i=p
\right).
\end{aligned}
\]

The key observation is that an inexperienced-compatible applicant satisfying
\[
\underline p_0^A(q)<p_i(q)<p^F(q)
\]
has zero probability of being hired under one-stage hiring but has a strictly
positive path to hiring under an optimal multistage policy. In particular, if
the vacancy remains open until the applicant becomes the firm's only remaining
option, Lemma \ref{lem:binary-assessment-cutoff} implies that the firm strictly
prefers to administer the intermediate assessment. A favorable assessment then
leads to full screening and, because the applicant is compatible, to hiring.
Proposition \ref{prop:optimal-ms-01-helped-app} in Appendix
\ref{sec:app_6} formalizes this argument and gives an explicit positive lower
bound on the applicant's hiring-probability gain.

To measure how many inexperienced-compatible applicants fall into this region,
define
\[
\nu_{01}(q)
:=
\Pr\left(
\underline p_0^A(q)
<
p_i(q)
<
p^F(q)
\mid
e_i=0,\ k_i=1,\ d_i=1
\right).
\]
Thus \(\nu_{01}(q)\) is the share of submitted inexperienced-compatible
applicants who are rejected under one-stage hiring but have a strictly positive
path to hiring under multistage hiring. Whenever
\[
\underline p_0^A(q)<p^F(q),
\]
the full support of the Gaussian materials signal implies
\[
\nu_{01}(q)>0.
\]
Moreover, every applicant in this region obtains a strictly positive
hiring-probability gain under optimal multistage hiring. Corollary
\ref{cor:random-01-helped-app} in Appendix \ref{sec:app_6} states this finite-$q$
result formally.

We now turn to the main result of this subsection. In the high-AI limit, the
share of inexperienced-compatible applicants who benefit from multistage
hiring can converge to one. Recall that
\[
\pi_0
=
\Pr(k_i=1\mid e_i=0)
=
\Phi\left(
\frac{\mu_0}{\sigma_z}
\right)
\]
is the inexperienced-group prior compatibility probability.

\begin{cor}[Almost all inexperienced-compatible applicants benefit in the high-AI limit]
\label{cor:high-ai-most-01-helped}
Let $\{q_n\}$ satisfy
\[
\sigma_m^2(q_n)\to\infty
\]
and
\[
p^F(q_n)\to p_\infty^F\in(0,1).
\]
Suppose
\[
\limsup_{n\to\infty}
\underline p_0^A(q_n)
<
\pi_0
<
p_\infty^F.
\]
Then
\[
\nu_{01}(q_n)
\to1.
\]
Consequently,
\[
\Pr\left(
G_{01}^*\left(p_i(q_n);q_n\right)>0
\mid
e_i=0,\ k_i=1,\ d_i=1
\right)
\to1.
\]
\end{cor}

As application materials become uninformative, the pre-assessment posterior of
a submitted inexperienced-compatible applicant converges to $\pi_0$. If
$\pi_0$ remains strictly above the assessment cutoff and strictly below the
one-stage screening cutoff, almost every submitted inexperienced-compatible
applicant is rejected under one-stage hiring but obtains a strictly positive
hiring-probability gain under optimal multistage hiring. Thus the benefit of the
additional screening stage is not confined to an exceptional set of applicants:
in the high-AI limit, it extends to almost the entire inexperienced-compatible
group.

When the high-noise screening-failure condition from Section
\ref{subsec:market-failure} also holds, i.e.,
\[
p_\infty^P<\pi_0<p_\infty^F,
\]
the intermediate stage benefits the same inexperienced-compatible applicants
who are socially worth screening but rejected under one-stage hiring.

\section{Conclusion}
\label{sec:conclusion}

AI-assisted job search changes the hiring market not only by affecting how many applications workers submit, but also by changing what firms can infer from those applications. In our model, less informative application materials lead firms to rely more heavily on applicants' experience. Inexperienced-compatible applicants are especially harmed as they lose the individualized evidence that could distinguish them from others with little experience. Because firms do not internalize workers' surplus from a successful match, they may also reject applicants who are socially worth screening. In very noisy, high-AI environments, these failures can become systematic: the firm may screen no one or screen only experienced applicants.

Multistage hiring can mitigate this problem. A sufficiently informative and inexpensive intermediate assessment allows firms to acquire new evidence of fit before committing to costly full screening. We identify conditions under which this option strictly increases firm payoffs and gives some inexperienced-compatible applicants rejected under one-stage hiring a positive-probability path to employment. As AI saturation becomes sufficiently high, this benefit extends to almost all such applicants. Expanding evaluation opportunities can therefore serve firms' interests as well as those of applicants otherwise excluded.

Our results highlight that easier access to applying is not the same as access to credible evaluation. Preserving access for high-potential workers without conventional experience therefore requires attention to the design of screening and hiring processes. As application signals deteriorate, multistage hiring processes that create additional opportunities to demonstrate fit can help restore market functioning by reopening evaluation and hiring opportunities that would otherwise disappear.

\bibliographystyle{plainnat}
\bibliography{job_search_refs}
\appendix
\section{Proofs for Section \ref{sec:characterization}}
\label{sec:app_4}

\begin{proof}[Proof of Proposition \ref{prop:bayesian-reweighting}]
Fix $q$ and a submitted applicant $i$ with experience $e_i=e$.
Because the application decision depends on $(e_i,\eta_i)$ and
$\eta_i$ is independent of $(z_i,\varepsilon_i)$ conditional on experience, application is
uninformative about match quality and the materials signal conditional
on experience. Thus,
\[
z_i\mid m_i,e_i=e,d_i=1;q
\overset{d}{=}
z_i\mid m_i,e_i=e;q.
\]

Under \eqref{eq:quality-prior} and \eqref{eq:materials-signal}, standard
Gaussian updating gives
\[
z_i\mid m_i,e_i=e;q
\sim
\mathcal N\left(
S_i(q),
\sigma_{\mathrm{post}}^2(q)
\right),
\]
where
\[
\sigma_{\mathrm{post}}^2(q)
=
\left(
\frac{1}{\sigma_z^2}
+
\frac{1}{\sigma_m^2(q)}
\right)^{-1}
=
\frac{\sigma_z^2\sigma_m^2(q)}
{\sigma_z^2+\sigma_m^2(q)}
\]
and
\begin{align*}
S_i(q)
&=
\sigma_{\mathrm{post}}^2(q)
\left(
\frac{\mu_e}{\sigma_z^2}
+
\frac{m_i}{\sigma_m^2(q)}
\right) \\
&=
\kappa(q)m_i+(1-\kappa(q))\mu_e,
\end{align*}
with
\[
\kappa(q)
=
\frac{\sigma_z^2}
{\sigma_z^2+\sigma_m^2(q)}.
\]

Let
\[
\Delta_\mu:=\mu_1-\mu_0>0.
\]
Since
\[
\mu_e=\mu_0+\Delta_\mu e,
\]
the posterior mean can equivalently be written as
\[
S_i(q)
=
(1-\kappa(q))\mu_0
+
\theta(q)e_i
+
\kappa(q)m_i,
\]
where
\[
\theta(q)
=
(1-\kappa(q))\Delta_\mu.
\]

Compatibility is the event $k_i=1$ if and only if $z_i\ge0$.
Therefore,
\begin{align*}
p_i(q)
&=
\Pr(k_i=1\mid m_i,e_i;q) \\
&=
\Pr(z_i\ge0\mid m_i,e_i;q) \\
&=
\Phi\left(
\frac{S_i(q)}
{\sigma_{\mathrm{post}}(q)}
\right).
\end{align*}
Because $\sigma_{\mathrm{post}}(q)>0$ and $\Phi$ is strictly increasing,
for any submitted applicants $i$ and $j$,
\[
p_i(q)\ge p_j(q)
\quad\Longleftrightarrow\quad
S_i(q)\ge S_j(q).
\]

Now consider $\widetilde q>q$. By Assumption
\ref{assump:sigma-increasing},
\[
\sigma_m^2(\widetilde q)>\sigma_m^2(q).
\]
It follows immediately that
\[
\kappa(\widetilde q)
=
\frac{\sigma_z^2}
{\sigma_z^2+\sigma_m^2(\widetilde q)}
<
\frac{\sigma_z^2}
{\sigma_z^2+\sigma_m^2(q)}
=
\kappa(q).
\]
Since $\Delta_\mu>0$,
\[
\theta(\widetilde q)
=
(1-\kappa(\widetilde q))\Delta_\mu
>
(1-\kappa(q))\Delta_\mu
=
\theta(q).
\]

Finally,
\begin{align*}
\frac{\theta(q)}{\kappa(q)}
&=
\Delta_\mu
\frac{1-\kappa(q)}{\kappa(q)} \\
&=
\frac{\Delta_\mu\,\sigma_m^2(q)}
{\sigma_z^2}.
\end{align*}
Assumption \ref{assump:sigma-increasing} therefore implies that
$\theta(q)/\kappa(q)$ is strictly increasing in $q$.
\end{proof}

\begin{proof}[Proof of Proposition \ref{prop:optimal-screening}]
Fix $q$ and a realized pool of $M$ submitted applications. Write
\[
V:=v_F
\qquad\text{and}\qquad
c:=c_s(q).
\]
By assumption, $V>0$. Conditional on the observed application
information, applicant $i$ is compatible with probability $p_i(q)$.
Because applicant primitives are independent across applicants,
compatibility realizations are independent conditional on the
firm's observed information.

We first establish the optimal ordering. Consider two applicants $i$
and $j$ who are screened consecutively, conditional on the vacancy
remaining open when the firm reaches them. The probability that at
least one of the two applicants is compatible, and the probability
that both applicants are incompatible, do not depend on their order.
Thus the expected hiring payoff and the continuation value after both
applicants fail are the same under either order.

If the firm screens $i$ before $j$, its expected screening cost over
these two applicants is
\[
c+(1-p_i(q))c.
\]
If it screens $j$ before $i$, the corresponding expected cost is
\[
c+(1-p_j(q))c.
\]
Hence the expected-payoff difference between placing $i$ before $j$
and placing $j$ before $i$ is
\[
c\bigl(p_i(q)-p_j(q)\bigr).
\]
Therefore,
\[
p_i(q)\ge p_j(q)
\]
implies that placing $i$ weakly before $j$ is optimal. Repeated
pairwise interchanges imply that applicants can be ordered so that
\[
p_1(q)\ge p_2(q)\ge\cdots\ge p_M(q).
\]

Now suppose the firm screens the first $L$ applicants in this order,
unless it discovers a compatible applicant earlier. Its expected
payoff is
\[
U_L(q)
=
V\left[
1-\prod_{j=1}^{L}(1-p_j(q))
\right]
-
c\sum_{\ell=1}^{L}
\prod_{j=1}^{\ell-1}(1-p_j(q)),
\]
where the empty product is equal to one. The first term is the firm's
expected hiring payoff, and the second is its expected total screening
cost.

The marginal payoff from adding applicant $L$ to the end of the
screening list is
\begin{align*}
U_L(q)-U_{L-1}(q)
&=
\prod_{j=1}^{L-1}(1-p_j(q))
\left[
Vp_L(q)-c
\right].
\end{align*}
The product preceding the bracket is the probability that applicant
$L$ is reached. It is nonnegative. Thus applicant $L$ is weakly worth
including if and only if
\[
Vp_L(q)\ge c.
\]
Substituting $V=v_F$ and $c=c_s(q)$ gives
\[
p_L(q)
\ge
\frac{c_s(q)}{v_F}
=
p^F(q).
\]

Because applicants are ordered by decreasing posterior compatibility,
once an applicant fails this inequality, every remaining applicant
also fails it. Hence the firm optimally screens applicants sequentially
in decreasing order of posterior compatibility until either it finds
a compatible applicant or no remaining applicant satisfies
\[
p_i(q)\ge p^F(q).
\]
At equality the firm is indifferent; the proposition adopts the
tie-breaking convention that an indifferent applicant is included in
the screening list.

Finally, under Assumption \ref{assump:cs-increasing}, for every
$\widetilde q>q$,
\[
c_s(\widetilde q)>c_s(q).
\]
Since $v_F>0$,
\[
p^F(\widetilde q)
=
\frac{c_s(\widetilde q)}{v_F}
>
\frac{c_s(q)}{v_F}
=
p^F(q).
\]
\end{proof}

\section{Proofs for Section \ref{sec:workforce-entry}}
\label{sec:app_5}

\begin{proof}[Proof of Proposition \ref{prop:mean-score-effects}]
Because applicants do not observe $z_i$ before applying and application costs
are independent of $(z_i,\varepsilon_i)$ conditional on experience,
conditioning on submission does not change the distribution of $(z_i,m_i)$
within an experience group. Hence
\[
\mathbb E[m_i\mid e_i=e,k_i=k,d_i=1]
=
\mathbb E[z_i\mid e_i=e,k_i=k]
=
\mu_{ek},
\]
where we use $\mathbb E[\varepsilon_i]=0$. Therefore,
\[
\bar S_{ek}(q)
=
\kappa(q)\mu_{ek}
+
(1-\kappa(q))\mu_e.
\]
For $q<\widetilde q$,
\begin{align*}
\Delta\bar S_{ek}(q,\widetilde q)
&=
\bar S_{ek}(\widetilde q)-\bar S_{ek}(q)\\
&=
\left[
\kappa(\widetilde q)-\kappa(q)
\right]
\left(
\mu_{ek}-\mu_e
\right).
\end{align*}

It remains to establish the ordering of these changes. Define
\[
t_e:=\frac{\mu_e}{\sigma_z}.
\]
Since $z_i\mid e_i=e\sim\mathcal N(\mu_e,\sigma_z^2)$ and compatibility is
the event $z_i\ge0$, the conditional means are
\[
\mu_{e1}
=
\mu_e
+
\sigma_z
\frac{\phi(t_e)}{\Phi(t_e)}
\]
and
\[
\mu_{e0}
=
\mu_e
-
\sigma_z
\frac{\phi(t_e)}{\Phi(-t_e)}.
\]
Thus
\[
\mu_{e1}-\mu_e
=
\sigma_z
\frac{\phi(t_e)}{\Phi(t_e)}
>0
\]
and
\[
\mu_{e0}-\mu_e
=
-\sigma_z
\frac{\phi(t_e)}{\Phi(-t_e)}
<0.
\]

The inverse Mills ratio
\[
\frac{\phi(t)}{\Phi(t)}
\]
is strictly decreasing in $t$, while
\[
\frac{\phi(t)}{\Phi(-t)}
\]
is strictly increasing in $t$. Since $\mu_1>\mu_0$, we have $t_1>t_0$ and
therefore
\[
\mu_{01}-\mu_0
>
\mu_{11}-\mu_1
>
0,
\]
while
\[
\mu_{10}-\mu_1
<
\mu_{00}-\mu_0
<
0.
\]

By Assumption \ref{assump:sigma-increasing},
\[
\kappa(\widetilde q)-\kappa(q)<0.
\]
Multiplying the preceding inequalities by this negative quantity gives
\[
\Delta\bar S_{01}(q,\widetilde q)
<
\Delta\bar S_{11}(q,\widetilde q)
<
0
<
\Delta\bar S_{00}(q,\widetilde q)
<
\Delta\bar S_{10}(q,\widetilde q),
\]
as required.
\end{proof}

\begin{proof}[Proof of Proposition
\ref{prop:two-adverse-shortlisting-channels}]
Fix an experience group $e$. Recall that
\[
X_i
=
\frac{z_i-\mu_e}{\sigma_z},
\qquad
Y_i(q)
=
\frac{m_i-\mu_e}
{\sqrt{\sigma_z^2+\sigma_m^2(q)}}.
\]
Conditional on $e_i=e$, $(X_i,Y_i(q))$ is standard bivariate normal with
correlation
\[
r(q)
=
\frac{\sigma_z}
{\sqrt{\sigma_z^2+\sigma_m^2(q)}}.
\]

For $a,b\in\mathbb R$ and $r\in(-1,1)$, let
\[
\Phi_2(a,b;r)
=
\Pr(X\le a,Y\le b)
\]
denote the standard bivariate normal CDF, and let
$\phi_2(a,b;r)$ denote its density. The standard derivative identity
\[
\frac{\partial}{\partial r}
\Phi_2(a,b;r)
=
\phi_2(a,b;r)>0
\]
implies
\[
\frac{\partial}{\partial r}
\Pr(X\ge a_e,Y\ge b)
=
\phi_2(a_e,b;r)>0.
\]
Since $\Pr(X\ge a_e)$ does not depend on $r$,
\[
\frac{\partial\psi_{e1}(b,r)}{\partial r}
=
\frac{
\phi_2(a_e,b;r)
}{
\Pr(X\ge a_e)
}
>0.
\]

Similarly,
\[
\Pr(X<a_e,Y\ge b)
=
\Pr(Y\ge b)
-
\Pr(X\ge a_e,Y\ge b).
\]
The first term is independent of $r$, so
\[
\frac{\partial\psi_{e0}(b,r)}{\partial r}
=
-
\frac{
\phi_2(a_e,b;r)
}{
\Pr(X<a_e)
}
<0.
\]

We next derive the relative standardized thresholds. By the definition of
$Y_i(q)$,
\[
m_i
=
\mu_e
+
\sqrt{\sigma_z^2+\sigma_m^2(q)}\,Y_i(q).
\]
Because
\[
\kappa(q)
=
\frac{\sigma_z^2}
{\sigma_z^2+\sigma_m^2(q)}
\]
and
\[
r(q)
=
\frac{\sigma_z}
{\sqrt{\sigma_z^2+\sigma_m^2(q)}},
\]
the posterior-mean score can be written as
\[
S_i(q)
=
\mu_e+\sigma_z r(q)Y_i(q).
\]
The screening condition
\[
p_i(q)\ge p^F(q)
\]
is equivalent to
\[
S_i(q)
\ge
\sigma_{\mathrm{post}}(q)
\Phi^{-1}\left(p^F(q)\right).
\]
Therefore it is equivalent to
\[
Y_i(q)\ge b_e(q),
\]
where
\[
b_e(q)
=
\frac{
\sigma_{\mathrm{post}}(q)
\Phi^{-1}\left(p^F(q)\right)-\mu_e
}{
\sigma_zr(q)
}.
\]
The term involving the posterior cutoff is common across experience groups, so
\[
b_0(q)-b_1(q)
=
\frac{\mu_1-\mu_0}{\sigma_zr(q)}
=
\frac{\Delta_\mu}{\sigma_zr(q)}.
\]
Differentiating gives
\[
\frac{d}{dq}
\left[
b_0(q)-b_1(q)
\right]
=
-
\frac{\Delta_\mu r'(q)}
{\sigma_zr(q)^2}.
\]
Assumption \ref{assump:sigma-increasing} implies $r'(q)<0$, and
$\Delta_\mu>0$, so
\[
\frac{d}{dq}
\left[
b_0(q)-b_1(q)
\right]>0.
\]
\end{proof}

\begin{proof}[Proof of Proposition \ref{prop:private-social-cutoffs}]
At the screening stage, application costs are sunk. A compatible hire generates
total surplus
\[
v_F+v_A.
\]
Thus the planner's screening problem has exactly the same structure as the
firm's problem in Proposition \ref{prop:optimal-screening}, except that the
value of discovering a compatible applicant is $v_F+v_A$ rather than
$v_F$.

The same pairwise-interchange argument therefore implies that the planner screens
applicants in decreasing order of posterior compatibility. Conditional on
reaching applicant $i$, the expected marginal social value of screening is
\[
(v_F+v_A)p_i(q)-c_s(q).
\]
Hence the planner screens applicant $i$ if and only if
\[
p_i(q)
\ge
\frac{c_s(q)}
{v_F+v_A}
=
p^P(q).
\]
Because $v_A>0$,
\[
p^P(q)
=
\frac{c_s(q)}{v_F+v_A}
<
\frac{c_s(q)}{v_F}
=
p^F(q).
\]
Thus an applicant is included in the planner's screening list but not the
firm's if and only if
\[
p^P(q)\le p_i(q)<p^F(q).
\]
\end{proof}

\begin{proof}[Proof of Proposition \ref{prop:no-hire-market-failure}]
Suppose
\[
p^P(q)<p^{\max}(q)<p^F(q).
\]
Since every submitted applicant satisfies
\[
p_i(q)\le p^{\max}(q)<p^F(q),
\]
Proposition \ref{prop:optimal-screening} implies that the firm's screening list
is empty.

On the other hand, an applicant attaining $p^{\max}(q)$ satisfies
\[
p^{\max}(q)>p^P(q).
\]
By Proposition \ref{prop:private-social-cutoffs}, the social planner includes
this applicant in its screening list. Hence the planner's screening list is
nonempty while the firm's is empty.
\end{proof}

\begin{proof}[Proof of Proposition
\ref{prop:inexperienced-compatible-market-failure}]
Suppose
\[
p^P(q)<p_{01}^{\max}(q)<p^F(q).
\]
Every submitted inexperienced-compatible applicant satisfies
\[
p_i(q)\le p_{01}^{\max}(q)<p^F(q),
\]
so none is included in the firm's screening list.

An applicant attaining $p_{01}^{\max}(q)$ satisfies
\[
p_{01}^{\max}(q)>p^P(q),
\]
so Proposition \ref{prop:private-social-cutoffs} implies that this applicant is
included in the planner's screening list. Hence at least one submitted
inexperienced-compatible applicant is socially worth screening but excluded by
the firm.
\end{proof}

\begin{proof}[Proof of Lemma \ref{lem:high-noise-posterior-limit}]
Fix an applicant with experience $e_i=e$ and realized compatibility
$k_i=k$. Recall that
\[
S_i(q)
=
\kappa(q)m_i+(1-\kappa(q))\mu_e
=
\mu_e+\kappa(q)(z_i-\mu_e)+\kappa(q)\varepsilon_i.
\]
Along a sequence $q_n$ such that
\[
\sigma_m^2(q_n)\to\infty,
\]
we have
\[
\kappa(q_n)
=
\frac{\sigma_z^2}
{\sigma_z^2+\sigma_m^2(q_n)}
\to0.
\]

Conditional on $(e_i=e,k_i=k)$, $z_i$ has a truncated normal distribution and
therefore has finite second moment. Hence
\[
\kappa(q_n)(z_i-\mu_e)\to0
\]
in probability.

Also,
\[
\operatorname{Var}
\left(
\kappa(q_n)\varepsilon_i
\right)
=
\kappa(q_n)^2\sigma_m^2(q_n)
=
\frac{
\sigma_z^4\sigma_m^2(q_n)
}{
\left(
\sigma_z^2+\sigma_m^2(q_n)
\right)^2
}
\to0.
\]
Therefore,
\[
\kappa(q_n)\varepsilon_i\to0
\]
in probability, and consequently
\[
S_i(q_n)\to\mu_e
\]
in probability conditional on $(e_i=e,k_i=k)$.

Moreover,
\[
\sigma_{\mathrm{post}}^2(q_n)
=
\frac{
\sigma_z^2\sigma_m^2(q_n)
}{
\sigma_z^2+\sigma_m^2(q_n)
}
\to
\sigma_z^2.
\]
Since
\[
p_i(q_n)
=
\Phi\left(
\frac{S_i(q_n)}
{\sigma_{\mathrm{post}}(q_n)}
\right),
\]
the continuous mapping theorem gives
\[
p_i(q_n)
\to
\Phi\left(
\frac{\mu_e}{\sigma_z}
\right)
=
\pi_e
\]
in probability.

Application is independent of $(z_i,m_i)$ conditional on experience, so the
same convergence holds conditional on $d_i=1$.

For any fixed finite submitted pool $\mathcal B$,
\[
\max_{i\in\mathcal B}
\left|
p_i(q_n)-\pi_{e_i}
\right|
\to0
\]
in probability by a union bound over the finitely many applicants.
\end{proof}

\begin{proof}[Proof of Corollary \ref{cor:high-noise-screening-failures}]
By Lemma \ref{lem:high-noise-posterior-limit}, conditional on the realized
types in the fixed finite submitted pool $\mathcal B$,
\[
\max_{i\in\mathcal B}
\left|
p_i(q_n)-\pi_{e_i}
\right|
\to0
\]
in probability.

For part (i), it follows that
\[
\max_{i\in\mathcal B}p_i(q_n)
\to
\pi_{\mathcal B}^{\max}
\]
in probability. By assumption,
\[
p_\infty^P
<
\pi_{\mathcal B}^{\max}
<
p_\infty^F,
\]
while
\[
p^P(q_n)\to p_\infty^P,
\qquad
p^F(q_n)\to p_\infty^F.
\]
The strict inequalities therefore imply
\[
\Pr\left(
p^P(q_n)
<
\max_{i\in\mathcal B}p_i(q_n)
<
p^F(q_n)
\right)
\to1.
\]
Proposition \ref{prop:no-hire-market-failure} then gives the no-hire screening
failure with probability approaching one.

For part (ii), every applicant in $\mathcal B_{01}$ has experience $e=0$.
Hence
\[
\max_{i\in\mathcal B_{01}}
\left|
p_i(q_n)-\pi_0
\right|
\to0
\]
in probability. If
\[
p_\infty^P<\pi_0<p_\infty^F,
\]
then
\[
\Pr\left(
p^P(q_n)
<
\min_{i\in\mathcal B_{01}}p_i(q_n)
\le
\max_{i\in\mathcal B_{01}}p_i(q_n)
<
p^F(q_n)
\right)
\to1.
\]
Thus every submitted inexperienced-compatible applicant lies below the firm's
cutoff and above the planner's cutoff with probability approaching one.

Finally, if $\mathcal B_1$ is nonempty and
\[
p_\infty^F<\pi_1,
\]
then
\[
\min_{i\in\mathcal B_1}p_i(q_n)\to\pi_1
\]
in probability. Combining this convergence with
\[
\max_{i\in\mathcal B_{01}}p_i(q_n)\to\pi_0
\]
and
\[
\pi_0<p_\infty^F<\pi_1
\]
gives
\[
\Pr\left(
\min_{i\in\mathcal B_1}p_i(q_n)
>
p^F(q_n)
>
\max_{i\in\mathcal B_{01}}p_i(q_n)
\right)
\to1.
\]
\end{proof}

\section{Proofs and Supporting Results for Section \ref{sec:multistage}}
\label{sec:app_6}

\begin{proof}[Proof of Lemma \ref{lem:binary-assessment-cutoff}]
Fix $q$, experience group $e$, and a pre-assessment posterior
\[
p<p^F(q).
\]
Write
\[
V:=v_F,
\qquad
p^F(q)=\frac{c_s(q)}{V}.
\]

The probability of a favorable assessment outcome is
\[
\Pr(A_i=H\mid p,e)
=
\tau_ep+\ell_e(1-p).
\]
Bayes' rule gives
\[
p_H^e(p)
=
\frac{\tau_ep}
{\tau_ep+\ell_e(1-p)}
\]
and
\[
p_L^e(p)
=
\frac{(1-\tau_e)p}
{(1-\tau_e)p+(1-\ell_e)(1-p)}.
\]
Because $\tau_e>\ell_e$,
\[
p_H^e(p)>p>p_L^e(p).
\]
Since $p<p^F(q)$, it follows immediately that
\[
p_L^e(p)<p^F(q),
\]
so the firm never proceeds to full screening following $L$ in the
one-applicant continuation problem.

Consider the policy that administers the intermediate assessment and proceeds
to full screening only following $H$. Its expected payoff is
\begin{align*}
&-c_\ell(q)
+
\Pr(A_i=H\mid p,e)
\left[
Vp_H^e(p)-c_s(q)
\right] \\
={}&
-c_\ell(q)
+
\left[
\tau_ep+\ell_e(1-p)
\right]
\left[
V
\frac{\tau_ep}
{\tau_ep+\ell_e(1-p)}
-c_s(q)
\right]\\
={}&
-c_\ell(q)
+
V\tau_ep
-
c_s(q)
\left[
\tau_ep+\ell_e(1-p)
\right].
\end{align*}
Substituting
\[
c_s(q)=Vp^F(q)
\]
gives
\[
-c_\ell(q)
+
V
\left[
p\Bigl(
\tau_e(1-p^F(q))
+
\ell_ep^F(q)
\Bigr)
-
\ell_ep^F(q)
\right].
\]
This payoff is nonnegative if and only if
\[
p
\ge
\frac{
c_\ell(q)/V+\ell_ep^F(q)
}{
\tau_e(1-p^F(q))+\ell_ep^F(q)
}
=
\underline p_e^A(q).
\]
It is strictly positive if and only if
\[
p>\underline p_e^A(q).
\]

For any $p\ge\underline p_e^A(q)$, the preceding inequality also implies that
a favorable assessment raises the posterior above the full-screening cutoff,
so screening following $H$ is optimal. Thus, among applicants satisfying
$p<p^F(q)$, the firm weakly prefers assessment if and only if
\[
p\ge\underline p_e^A(q),
\]
with strict preference when
\[
\underline p_e^A(q)<p<p^F(q).
\]

Finally, there exist below-cutoff applicants for whom assessment is strictly
preferred if and only if
\[
\underline p_e^A(q)<p^F(q).
\]
Substituting the definition of $\underline p_e^A(q)$ and rearranging,
\begin{align*}
\underline p_e^A(q)<p^F(q)
\quad\Longleftrightarrow\quad
\frac{c_\ell(q)}{V}
&<
p^F(q)(1-p^F(q))(\tau_e-\ell_e).
\end{align*}
Multiplying by $V=v_F$ gives
\[
c_\ell(q)
<
v_Fp^F(q)(1-p^F(q))
(\tau_e-\ell_e),
\]
as required.
\end{proof}

\begin{proof}[Proof of Corollary \ref{cor:single-firm-adopts-ms}]
Suppose that for some experience group $e$,
\[
\lambda_e\alpha_e(q)>0
\]
and
\[
c_\ell(q)
<
v_Fp^F(q)(1-p^F(q))
(\tau_e-\ell_e).
\]
By Lemma \ref{lem:binary-assessment-cutoff},
\[
\underline p_e^A(q)<p^F(q).
\]

Conditional on $e_i=e$ and $d_i=1$, the application-materials signal has full
support on $\mathbb R$. Moreover, posterior compatibility
\[
p_i(q)
=
\Pr(k_i=1\mid m_i,e_i=e;q)
\]
is continuous and strictly increasing in $m_i$, with limits zero and one as
$m_i\to-\infty$ and $m_i\to\infty$, respectively. Hence
\[
\Pr\left(
\underline p_e^A(q)
<
p_i(q)
<
p^F(q)
\mid
e_i=e,d_i=1
\right)
>0.
\]

Because $\lambda_e\alpha_e(q)>0$, there is positive probability that the
submitted pool contains such an applicant. Consider a feasible multistage policy
that first replicates the optimal one-stage policy. If the vacancy remains open
after the one-stage screening list is exhausted, the policy administers the
intermediate assessment to one such below-cutoff applicant.

Conditional on reaching this applicant, Lemma
\ref{lem:binary-assessment-cutoff} implies that the assessment has strictly
positive expected value. There is positive probability that the vacancy remains
open until this point; for example, the event that all other potential
applicants are incompatible has strictly positive probability. Therefore this
feasible policy weakly improves on the optimal one-stage policy at every history
and strictly improves on it with positive probability.

It follows that
\[
\mathbb E\left[
U^{MS,*}(q)-U^{1S,*}(q)
\right]
>0.
\]
If an ex ante optimal multistage policy never administered the intermediate
assessment, it would be a feasible one-stage policy and could achieve at most
the one-stage value. Hence every ex ante optimal multistage policy must
administer the intermediate assessment with positive probability.
\end{proof}

\begin{proof}[Proof of Corollary \ref{cor:multistage-avoids-failures}]
For part (i), suppose the one-stage process exhibits the no-hire screening
failure. Then every submitted applicant satisfies
\[
p_i(q)<p^F(q),
\]
and the optimal one-stage payoff is
\[
U^{1S,*}(q)=0.
\]
Suppose there exists a submitted applicant $i$ satisfying
\[
\max\left\{
p^P(q),\underline p_{e_i}^A(q)
\right\}
<
p_i(q)
<
p^F(q).
\]
The firm can ignore all other applicants and treat $i$ as its only remaining
option. By Lemma \ref{lem:binary-assessment-cutoff}, administering the
intermediate assessment then has strictly positive expected payoff. Hence there
exists a feasible multistage policy with strictly positive payoff, so
\[
U^{MS,*}(q)>0=U^{1S,*}(q).
\]
Any optimal multistage policy must therefore administer an intermediate
assessment with positive probability. Because a favorable assessment outcome
has positive probability and raises the posterior above $p^F(q)$, the firm also
proceeds to full screening with positive probability.

For part (ii), suppose the one-stage process exhibits the
inexperienced-compatible screening failure and
\[
p_{01}^{\max}(q)>\underline p_0^A(q).
\]
Choose an applicant $i\in\mathcal A_{01}(q)$ attaining
$p_{01}^{\max}(q)$. By the definition of the one-stage failure,
\[
p^P(q)<p_i(q)<p^F(q).
\]
Together with the assumed inequality,
\[
\underline p_0^A(q)<p_i(q)<p^F(q).
\]
At any history at which the vacancy remains open and $i$ is the firm's only
remaining option, Lemma \ref{lem:binary-assessment-cutoff} implies that the firm
strictly prefers to administer the intermediate assessment. Following $H$, the
posterior exceeds $p^F(q)$ and the firm proceeds to full screening.
\end{proof}

\begin{prop}[Hiring gain for below-cutoff inexperienced-compatible applicants]
\label{prop:optimal-ms-01-helped-app}
Fix $q$ and suppose
\[
\underline p_0^A(q)<p<p^F(q).
\]
Then
\[
\Pr\left(
h_i^{MS,*}=1
\mid
e_i=0,\ k_i=1,\ d_i=1,\ p_i(q)=p
\right)
\ge
\tau_0\chi^{n-1},
\]
where
\[
\chi
:=
\Pr(k_i=0)
=
\sum_{e\in\{0,1\}}\lambda_e
\Phi\left(-\frac{\mu_e}{\sigma_z}\right)
>0.
\]
By contrast,
\[
\Pr\left(
h_i^{1S}=1
\mid
e_i=0,\ k_i=1,\ d_i=1,\ p_i(q)=p
\right)
=
0.
\]
Consequently,
\[
G_{01}^*(p;q)
\ge
\tau_0\chi^{n-1}
>
0.
\]
\end{prop}

\begin{proof}
Fix a submitted inexperienced-compatible applicant $i$ satisfying
\[
\underline p_0^A(q)
<
p_i(q)
<
p^F(q).
\]
Under one-stage hiring, Proposition \ref{prop:optimal-screening} implies that
$i$ is not included in the firm's screening list. Hence
\[
\Pr\left(
h_i^{1S}=1
\mid
e_i=0,k_i=1,d_i=1,p_i(q)
\right)
=
0.
\]

Now consider the multistage process. Let
\[
E_{-i}
:=
\left\{
k_j=0
\text{ for every }j\ne i
\right\}
\]
be the event that all other $n-1$ potential applicants are incompatible.
Applicant primitives are independent across applicants, so conditional on the
tagged applicant's type, submission decision, and posterior,
\[
\Pr(E_{-i})
=
\chi^{n-1}.
\]

On $E_{-i}$, no other applicant can fill the vacancy. Because retaining an
unassessed application is costless, applicant $i$ can remain available while
the firm evaluates its other options. If the vacancy remains open until $i$ is
the firm's only remaining option, Lemma \ref{lem:binary-assessment-cutoff}
implies that administering the intermediate assessment has strictly positive
value. Thus an optimal multistage policy cannot terminate search while $i$
remains available.

Conditional on $k_i=1$, applicant $i$ receives a favorable intermediate
assessment with probability
\[
\tau_0.
\]
Following $H$, the applicant's posterior exceeds $p^F(q)$, so full screening is
strictly profitable. Screening reveals that $k_i=1$, and the firm hires the
applicant. Therefore,
\[
\Pr\left(
h_i^{MS,*}=1
\mid
e_i=0,k_i=1,d_i=1,p_i(q)
\right)
\ge
\tau_0\chi^{n-1}.
\]
Subtracting the zero one-stage hiring probability gives
\[
G_{01}^*(p_i(q);q)
\ge
\tau_0\chi^{n-1}
>
0.
\]
\end{proof}

\begin{cor}[Positive share of inexperienced-compatible applicants helped]
\label{cor:random-01-helped-app}
Suppose $\sigma_m^2(q)>0$ and
\[
\underline p_0^A(q)<p^F(q).
\]
Then
\[
\nu_{01}(q)>0.
\]
Moreover,
\[
\Pr\left(
G_{01}^*(p_i(q);q)>0
\mid
e_i=0,\ k_i=1,\ d_i=1
\right)
\ge
\nu_{01}(q)
>
0.
\]
\end{cor}

\begin{proof}
Suppose
\[
\sigma_m^2(q)>0
\]
and
\[
\underline p_0^A(q)<p^F(q).
\]
Conditional on
\[
e_i=0,\qquad k_i=1,\qquad d_i=1,
\]
the materials signal $m_i$ has a continuous density that is strictly positive
on $\mathbb R$. Moreover,
\[
p_i(q)
=
\Pr(k_i=1\mid m_i,e_i=0;q)
\]
is a continuous and strictly increasing function of $m_i$ whose range is
$(0,1)$. Therefore the conditional distribution of $p_i(q)$ places positive
probability on every open subinterval of $(0,1)$.

In particular,
\[
\nu_{01}(q)
=
\Pr\left(
\underline p_0^A(q)
<
p_i(q)
<
p^F(q)
\mid
e_i=0,k_i=1,d_i=1
\right)
>0.
\]
For every posterior in this interval, Proposition
\ref{prop:optimal-ms-01-helped-app} implies
\[
G_{01}^*(p_i(q);q)>0.
\]
Hence
\[
\Pr\left(
G_{01}^*(p_i(q);q)>0
\mid
e_i=0,k_i=1,d_i=1
\right)
\ge
\nu_{01}(q)
>
0.
\]
\end{proof}

\begin{proof}[Proof of Corollary \ref{cor:high-ai-most-01-helped}]
By Lemma \ref{lem:high-noise-posterior-limit}, conditional on
\[
e_i=0,\qquad k_i=1,\qquad d_i=1,
\]
we have
\[
p_i(q_n)\to\pi_0
\]
in probability.

By assumption,
\[
\limsup_{n\to\infty}
\underline p_0^A(q_n)
<
\pi_0
<
p_\infty^F,
\]
and
\[
p^F(q_n)\to p_\infty^F.
\]
Therefore there exists $\varepsilon>0$ such that for all sufficiently large
$n$,
\[
\underline p_0^A(q_n)
<
\pi_0-\varepsilon
<
\pi_0+\varepsilon
<
p^F(q_n).
\]
Consequently,
\begin{align*}
\nu_{01}(q_n)
&=
\Pr\left(
\underline p_0^A(q_n)
<
p_i(q_n)
<
p^F(q_n)
\mid
e_i=0,k_i=1,d_i=1
\right)\\
&\ge
\Pr\left(
|p_i(q_n)-\pi_0|<\varepsilon
\mid
e_i=0,k_i=1,d_i=1
\right)
\to1.
\end{align*}
Thus
\[
\nu_{01}(q_n)\to1.
\]

Finally, Proposition \ref{prop:optimal-ms-01-helped-app} implies
\[
G_{01}^*(p_i(q_n);q_n)>0
\]
whenever
\[
\underline p_0^A(q_n)
<
p_i(q_n)
<
p^F(q_n).
\]
Therefore,
\[
\Pr\left(
G_{01}^*(p_i(q_n);q_n)>0
\mid
e_i=0,k_i=1,d_i=1
\right)
\ge
\nu_{01}(q_n)
\to1.
\]
\end{proof}

\end{document}